\documentclass{article} 

\usepackage{eurosym}
\usepackage{amssymb,amsmath,amsthm,bbm}
\usepackage{graphicx}
\newtheorem{Theorem}{Theorem}
\newtheorem{Proposition}{Proposition}
\newtheorem{Lemma}{Lemma}
\newtheorem{Corollary}{Corollary}

\theoremstyle{definition}
\newtheorem{Definition}{Definition}
\newtheorem{Assumption}{Assumption}
\theoremstyle{remark}
\newtheorem{Remark}{Remark}

\title{Price formation and optimal trading in intraday electricity markets with a major player}

\author{Olivier F{\'e}ron\footnote{Electricit\'e de France} \and Peter Tankov\footnote{CREST, ENSAE, Institut Polytechnique de Paris} \and Laura Tinsi$^{*,\dag}$}
\date{}

\begin{document}

\maketitle

\begin{abstract}We study price formation in intraday electricity markets in the presence of intermittent renewable generation. We consider the setting where a major producer may interact strategically with a large number of small producers.  Using stochastic control theory we identify the optimal strategies of agents with market impact and exhibit the Nash equilibrium in closed  form in the asymptotic framework of mean field games with a major player.
  This is a companion paper to \cite{feron2020price}, where a similar model is developed in the setting of identical agents.
\end{abstract}  



\section{Introduction}

The structure of electricity markets around the world has been profoundly transformed by the push towards liberalization in the late 90s and more recently by the massive arrival of renewable energy production. Distribution has been separated from production, and whereas in the past, a single producer could own the entire generation capacity of a given country or region, now a patchwork of small, often renewable, generators competes with a big historical producer.

The aim of this paper is to develop an equilibrium model for intraday electricity markets where a big producer with a significant market share competes with a large number of small renewable producers. Both the large producer and the small producers use the intraday markets to compensate their production and demand forecast errors, creating feedback effects on the market price. The large producer can act strategically, anticipating the impact of its decisions on the market prices and thus on the behavior of the small agents. The small agents are not strategic, and each one has a negligible effect on the market, however the behavior of all small agents taken together has a significant market impact. The large player has the first-mover advantage but does not observe the forecast of the minor players. These in turn have the information advantage since they observe the forecast of the major player as well as their own forecast. This leads to a stochastic leader-follower game where players interact through  the market price. We place ourselves in the linear-quadratic setting, exhibit the unique Nash equilibrium for this game in closed form in the framework of mean field games with a major player, and provide explicit formulas for the market price and the strategies of the agents. For a game with a finite number of players, we show how an $\varepsilon$-Nash equilibrium can be constructed from the mean field game solution.

This paper is a companion paper to \cite{feron2020price}, where a similar model is developed for the case of identical agents with symmetric interactions, and we refer the readers to that paper for a detailed review of literature on the stochastic and econometric modeling of intraday electricity markets. Here we simply mention that a similar linear-quadratic setting with linear market impact, has been used to determine optimal strategies for a single energy producer by 
A\"{\i}d, Gruet and Pham \cite{aid2016optimal} and Tan and Tankov \cite{tan2018optimal},  while Bouchard et al.~\cite{bouchard2018equilibrium} found an analytic expression for the equilibrium price in a linear-quadratic model of the stock market with symmetric interactions and perfect information. {We  also mention the recent paper of Aïd, Cosso and Pham \cite{aid2020equilibrium} where an equilibrium in complete information setting for a finite number of agents is derived in the intraday electricity market. This paper is close in spirit to the complete information framework of \cite{feron2020price}, but allows to treat the case of heterogeneous agents in conditions of uncertain production with possible outages and uncertain demand.} However, the complete information setting, where each agent observes all other agents' forecasts does not seem realistic in electricity markets.  The incomplete information setting, where each agent only observes its own forecast and the aggregate forecast, may  not be tractable for a finite number of agents.  Nevertheless, in \cite{feron2020price}, it has been shown that explicit solutions may be found in the mean field limit, where the number of agents is sent to infinity, and the influence of every single agent on the entire market becomes negligible.

The mean field games (MFG) are stochastic differential games with infinitely many players and symmetric interactions. The seminal papers of Lasry and Lions \cite{lasry2007mean}  and Huang, Caines, and Malhamé \cite{huang2006large} characterized the Nash equilibrium in this framework through a coupled system of a Hamilton-Jacobi-Bellman (HJB) and a Fokker-Planck (FP) equation.  Carmona and Delarue \cite{carmona2018probabilistic} developed an alternative probabilistic approach inspired by the Pontryagin principle and related the mean field game solution to a McKean-Vlasov Forward Backward Stochastic Differential Equation (FBSDE). 
The asymptotic results obtained in mean field games can be used to construct approximate equilibria ($\varepsilon$-Nash equilibria) for games with a finite number of players. Alternatively, equilibria of $N$-player games can be shown to converge to the corresponding weak mean field equilibria \cite{lacker2020}.

{While the original MFG setting involves symmetric agents, Huang \cite{huang2010large} introduced linear-quadratic mean field games with a major player. Caines and Nourian \cite{nourian2013} developed this approach in a general framework. In both papers, the mean field is exogenous to the actions of the major player. In contrast to these two papers, Bensoussan, Chau, and Yam in \cite{bensoussan2016mean},  and Carmona and Zhu in \cite{carmona2016probabilistic}, considered the endogenous case where the major player can influence the mean field. In \cite{bensoussan2016mean}, this leads to a leader-follower setting, also known as Stackelberg game. The authors derived a HJB equation and a FP equation to characterize the solution in the general case, while the linear quadratic setting was tackled with a stochastic maximum principle approach.  
More recently, Lasry and Lions \cite{lasry2018mean}, introduced a master equation accounting for this kind of major player model. Cardaliaguet, Cirant and Poretta \cite{cardaliaguet2020remarks} showed that the two previous approaches (\cite{lasry2018mean} and \cite{carmona2016probabilistic}) lead to the same Nash equilibria.} 

Financial markets and energy systems with many small interacting agents are a natural domain of applications of MFG. Casgrain and Jaimungal \cite{casgrain2020mean} applied the MFG theory  to optimal trade execution with price impact and terminal inventory liquidation condition, Fujii and Takahashi \cite{fujii2020mean}, 
used this theory to find an equilibrium price under market clearing conditions. In  \cite{casgrain2020mean, fujii2020mean}, the authors used the extended mean field setting to deal with  heterogeneous sub-populations of agents and incomplete information for \cite{casgrain2020mean}.  Alasseur, Tahar and Matoussi \cite{alasseur2017extended} developed a model for the optimal management of energy storage and distribution in a smart grid system through an extended MFG. Shrivats, Firoozi and Jaimungal \cite{shrivats2020mean} recently applied the theoretical setting developed in \cite{casgrain2020mean} to the case of trading in solar renewable energy certificate markets. Financial markets with a major player, leader-follower interactions and terminal inventory constraint were recently analyzed in \cite{fu2018mean,evangelista2020finite}. In \cite{fu2018mean}, the authors consider a Brownian filtration, impose a zero terminal inventory constraint and characterize the equilibrium in terms of a McKean-Vlasov FBSDE. In \cite{evangelista2020finite}, the authors study a market with a finite number of small players and a major player with first-mover advantage and information asymmetry, and characterize the solution in terms of a McKean-Vlasov FBSDE in a more general setting than that of \cite{fu2018mean}.

Among the cited papers, our methods and findings are closest in spirit to \cite{bensoussan2016mean,evangelista2020finite,fu2018mean}. 
Compared to the article \cite{bensoussan2016mean}, which of course solves a more general problem, without focusing on a specific application, our paper allows a much more general dynamics for the driving processes (general semimartingales) and does not require an a priori bound on the strategies to prove the existence of the Nash equilibrium in the presence of a major player.  Unlike the articles \cite{evangelista2020finite,fu2018mean}, which also study leader-follower games in financial markets, we consider a stochastic terminal constraint, characterize the equilibrium in explicit form, and show how an $\varepsilon$-Nash equilibrium for the finite-player game may be constructed from a mean field game solution.

The paper is organized as follows.  In section 2 we introduce the model and briefly recall the mean field game solution obtained in \cite{feron2020price} in the case of identical agents. In section 3, we present the main results of this paper in the setting allowing for the presence of a major player, whose influence on the market is not negligible in the MFG limit. In section 4, we show how the limiting MFG solution may be used to construct an approximate Nash equilibrium in a Stackelberg game with one major player and $N$ minor players. 
Finally, in section 5, equilibrium price trajectories, and the effect of market parameters on the price characteristics are illustrated with simulated data.

\section{Preliminaries}\label{homogeneous}
In this paper, we place ourselves in the intraday market for a given delivery hour starting at time $T$, where time $0$ corresponds to the opening time of the market (in EPEX Intraday this happens at 3PM on the previous day).   In reality, trading stops a few minutes before delivery time (e.g. 5 minutes for Germany). However, for the sake of simplicity we assume that market participants can trade during the entire period $[0, T]$. In the market, there are agents (producers or consumers) who are assumed to have taken a position in the day-ahead market and use the intraday market to manage the volume risk associated to the imperfect demand/production forecast. These forecasts represent the best estimate of the additional demand compared to the position taken by the agent in the spot market: to avoid imbalance penalties, the intraday position of the agent at the delivery date must therefore be equal to the realized demand, or, in other words, the last observed value of the demand forecast.

We consider the case of a Stackelberg game where an agent called "major agent" faces a large number of smaller agents called "minor agents".  We directly place ourselves in the setting of mean field games with a major player, that is, we assume that the number of small agents in the market is infinite, and the influence of each small agent on the market is negligible. The aggregate impact of the minor agents on the market is therefore modelled through a mean field.

Each agent observes the common national demand forecast, and the demand forecast of the major player. In addition, the small agents
also observe their individual demand forecasts, which are not observed by the other agents.
The common filtration of the market thus contains the information about the forecast of the major player and the common part of the forecasts of the minor players, but the small agents benefit from a private information advantage compared to the major player.

The demand forecast process and the position of the generic minor agent are given, respectively, by $X :=(X_t)_{0\leq t\leq T}$ and $\phi:=(\phi_t)_{0\leq t\leq T}$, while the forecast process and the position of the major agent are given, respectively, by $(X^0_t)_{t \in [0,T]}$ and $(\phi^0_t)_{t \in [0,T]}$. Note that the position and forecast of the minor and major agents are not expressed in the same units. Indeed, in the mean field game limit considered in this paper, we assume that the market is very large, so that the position of every minor agent compared to the market size is negligible, but the major agent takes up a nonzero share of the market, so that $\phi
^0$ and $X^0$ denote the position and forecast of the major agent normalized by the market size.

We denote by $\mathbb F$ the filtration which contains all information available to the generic minor agent and by $\mathbb F^0$ the filtration which contains all information available to the major agent. This filtration contains the information about the fundamental price, the information about the demand forecast of the major agent, and potentially some information about the demand forecast of the generic minor agent (the common noise) but, in general, not the full individual demand forecast of the generic agent.

Throughout the paper and for any $\mathbb{F}$-adapted process $(\zeta_t)_{t \in [0,T]}$, we will denote $\bar \zeta_t = \mathbb{E}[\zeta_t |\mathcal{F}^0_t] = \int_{\mathbb{R}}x \mu_t^{\zeta}(dx)$ where:   $\mu^{\zeta}_t:= \mathcal{L}(\zeta_t|\mathcal{F}^0_t)$. In view of the convergence results of \cite{feron2020price}, the (normalized) aggregate position of all minor agents is given by the expectation of $\phi$ with respect to the common noise: $\bar \phi_t = \mathbb E[X_t|\mathcal F
^0_t]$.  

We assume that the market price $(P_t)_{t \in [0,T]}$ is given by the fundamental price $(S_t)_{t \in [0,T]}$ plus a weighted combination of the aggregate position of the minor agents and the position of the major agent:   $$P_t = S_t + {a}\bar \phi_t+ {a}^0\phi^0_t, \quad \forall t \in [0,T]$$
{where ${a}$, ${a}^0$ are positive weights, which reflect the size of the major agent relative to the combined size of all minor agents and the overall strength of the market impact.}
Thus, the impact of each minor agent on the entire market is negligible, but the aggregate position of all minor agents, and the position of the major agent both have a nonzero impact.

 We say that the strategy of the generic minor agent $(\dot\phi_t)_{t \in [0,T]}$ is admissible if it is $\mathbb F$-adapted and square integrable. Similarly, the strategy of the major agent $(\dot\phi^0_t)_{t \in [0,T]}$ is admissible if it is $\mathbb F^0$-adapted and square integrable. The instantaneous cost of trading for the major agent and for the generic minor agent are defined, respectively, by:
\begin{equation}
    \dot \phi_t^0 P_t + \frac{\alpha_0(t)}{2} (\dot\phi^0_t)^2,\quad \text{and}\quad \dot \phi_t P_t + \frac{\alpha(t)}{2} (\dot\phi_t)^2, \quad  \forall t \in [0,T]
\end{equation}
In both instantaneous costs, the first term represents the actual cost of buying the electricity,
and the second term represents the cost of trading, where $\alpha(.)$ and  $\alpha_0(.)$ are continuous strictly positive functions on [0, T] reflecting the variation of market liquidity at the approach of
the delivery date. 

The objective function of the minor agent has the following form:\label{game}
\begin{equation}\label{minorobj}
    J^{MF}(\phi,\bar \phi, \phi^0) :=  -\mathbb E \left[\int_0^T  \frac{\alpha(t)}{2}\dot \phi^2_t + (S_t + {a}\bar \phi_t+ {a}^0\phi^0_t)\dot \phi_t dt + \frac{\lambda}{2} (\phi_T-X_T)^2\right],
  \end{equation}
while the objective function of the major agent writes,
\begin{equation}\label{majorobj}
     \quad J^{MF,0}(\phi^0,\bar \phi) :=    - \mathbb E \left[\int_0^T  \frac{\alpha_0(t)}{2}\dot {\phi^0_t}^2 +(S_t+ {a}\overline \phi_t + {a}^0\phi^0_t)\dot \phi^0_t dt + \frac{\lambda_0}{2} (\phi^0_T-X^0_T)^2\right].
\end{equation}
  
{Note that this formulation implies (as it is the case in real markets) that the major agent pays a much lower trading cost per unit traded and a much lower imbalance penalty than the minor agents.  Indeed, if the major agent paid the same quadratic cost/penalty as the minor agents, since the position of the major agent is very large, the quadratic trading cost/penalty would grow much faster than the linear part (the middle term in the formula), and the limiting formula would be degenerate, in the sense that the trading strategy would be independent from the price. To obtain a nondegenerate expression in terms of the normalized trading strategy of the major agent, we must therefore assume that the actual trading cost and penalty are also renormalized. The quantities $\alpha_0(t)$ and $\lambda_0$ are thus different from $\alpha(t)$ and $\lambda$ since they are of different nature: $\alpha(t)$ and $\lambda$ apply to the actual strategy of the generic agent, while $\alpha_0(t)$ and $\lambda_0$ apply to the normalized strategy of the major agent. The different nature of trading costs for minor and major agents is confirmed by other authors \cite{donier2015fully}: while the minor agents post their orders immediately in the order book, the major agent splits its orders into many small chunks to minimize trading costs. 
}

To close this introductory section, we briefly recall the main result from \cite{feron2020price}, which characterizes the mean field equilibrium in the setting of identical agents, in other words, we assume that $a_0 = 0$ until the end of this section.  
\begin{Definition}[mean field equilibrium]\label{defMFE} 
An admissible strategy  $\dot \phi^* := (\dot \phi^*_t)_{t \in [0,T]}$ is a mean field equilibrium in the setting of identical agents if it maximizes the functional \eqref{minorobj} with $a_0 = 0$ and satisfies $\bar \phi = \bar \phi^*$.
\end{Definition}

We make the following assumption. 
\begin{Assumption}\label{mfgass}${}$
\begin{itemize}
\item The process $S$ is square integrable and adapted to the filtration $\mathbb F^0$. 
\item The process $X$ is a square integrable martingale with respect to the filtration $\mathbb F$. 
\item The process $\overline X$ defined by $\overline X_t := \mathbb E[X_t|\mathcal F^ 0_t]$ for $0\leq t\leq T$ is a square integrable martingale with respect to the filtration $\mathbb F$.
\end{itemize}
\end{Assumption}
Note that if $X$ is an $\mathbb F$-martingale, then $\overline X$ is by construction an $\mathbb F^0$-martingale, but it may not necessarily be a martingale in the larger filtration $\mathbb F$.

The following theorem characterizes the mean field equilibrium in the identical agent setting.
In the theorem, we decompose the individual demand forecast as follows: $X_t = \overline{X}_t + \check{X}_t$, where $\overline{X}_t = \mathbb{E}\left[X_t |\mathcal{F}_t^0 \right]$, and we use the following shorthand notation:
\begin{align}
  \Delta_{s,t}&:= \int_s^t \frac{\eta(u,t)}{\alpha(u)} du\quad \text{with}\quad \eta(s,t) = e^{-\int_s^t \frac{a}{\alpha(u)}du}\quad \text{and}\quad \widetilde \Delta_{s,t} :=\int_s^t \alpha^{-1}(u) du\notag\\
   I_t&:= \int_0^t \frac{\eta(s,t)}{\alpha(s) } S_s ds, \quad \widetilde I_t := \mathbb E\left[\int_0^T\frac{\eta(s,T)}{\alpha(s) } S_s ds\Big|\mathcal F_t\right].
\end{align}

\begin{Theorem}\label{theorem_mfg}
Under Assumption \ref{mfgass}, the unique mean field equilibrium in the setting of identical agents is given by 
\begin{align}\label{strathomo}
\phi^{*}_t &= - I_t + \lambda\left[\Delta_{0,t}\frac{\widetilde I_0 +  \overline X_0}{1+\lambda \Delta_{0,T}}  + \int_0^t \Delta_{s,t}\frac{d\widetilde I_s+ d\overline X_s}{1+\lambda \Delta_{s,T}}\right.\\ \notag &\left.+\widetilde \Delta_{0,t}\frac{\check X_0}{1+\lambda\widetilde\Delta_{0,T}}+\int_0^t \widetilde \Delta_{s,t} \frac{ d\check X_s }{1+
  \lambda\widetilde\Delta_{s,T}} \right]. 
\end{align}
The equilibrium price has the following form:
\begin{equation}
\label{eq:price_incomplete_information}
P_t = S_t - aI_t + a\lambda\left[\Delta_{0,t}\frac{\widetilde I_0 +  \overline X_0}{1+\lambda \Delta_{0,T}}  + \int_0^t \Delta_{s,t}\frac{d\widetilde I_s+ d\overline X_s}{1+\lambda \Delta_{s,T}}\right].
\end{equation}
\end{Theorem}

\section{A game of a major and minor agents}\label{extension}

In this section we proceed to characterize the Nash equilibrium in the Stackelberg mean field game with a major player.   
Since a single minor agent has an infinitesimal impact on the market and cannot influence the mean field or the strategy of the major agent, the problem of the generic minor agent is to maximize $J^{MF}(\phi, \bar \phi, \phi^0)$ for fixed $\bar \phi$ and $\phi^0$. On the other hand, by modifying her strategy $\phi^0$, the major agent may influence the strategies of the minor agents, and thus also the mean field $\bar \phi$. This leads to the following definition of mean field equilibrium. 
\begin{Definition}[Stackelberg mean field equilibrium]\label{stackdef}
We call the triple $\phi^*$, $\bar\phi^*$, $\phi^{0*}$ Stackelberg mean field equilibrium for the game with a major and minor players if the following holds:
\begin{itemize}
    \item[i.] $\phi^*$ and $\phi^{0*}$ are admissible strategies for, respectively, the representative minor and the major players, the consistency condition $\bar \phi^*_t = \mathbb E[\phi^*_t|\mathcal F_t]$ is satisfied for all $t\in [0,T]$ and for any other admissible strategy for the representative minor player $\phi$, 
    $$
    J^{MF}(\phi,\bar\phi^*,\phi^{0*})\leq J^{MF}( \phi^*,\bar\phi^*,\phi^{0*})
    $$
    \item[ii.] For any other triple $( \phi,\bar\phi,\phi^0)$ satisfying condition i.,  
    \begin{align}
    J^{MF,0}( \phi^{0},\bar\phi)\leq J^{MF,0}( \phi^{0*},\bar\phi^*).\label{majopt}
    \end{align}
\end{itemize}
\end{Definition}

\begin{Assumption}\label{stackass}\text{}
In addition to assumption \ref{mfgass}, we also assume that the process $X^0$ is a square integrable martingale with respect to the filtration $\mathbb F^0$. 
\end{Assumption}

We start with the characterization of the optimal strategy for the minor agent.

\begin{Proposition}[Minor representative agent]\label{minor}Let $\bar \phi$ and $\phi^0$ be fixed.
The minor agent strategy $\phi$ maximizes  \eqref{minorobj} over the set of admissible strategies if and only if: 
\begin{equation}\label{minoreq}
    \dot \phi_t = -\frac{Y_t+S_t+a\overline{\phi}_t+a^0\phi^0_t}{\alpha(t)}, \quad \forall t \in [0,T],
\end{equation}
where $Y$ is a $\mathbb F$-martingale satisfying $Y_T = \lambda(\phi_T-X_T)$.
\end{Proposition}
\begin{proof}
The proof follows from the first step of the proof of Theorem \ref{theorem_mfg} (see Theorem 7 in \cite{feron2020price}) taking $\widetilde S = S + a_0\phi^0$ as fundamental price instead of $S$. 
\end{proof}

The problem of the major agent is more complex since the minor agents observe her actions and modify their strategies accordingly, which means that the mean field $\bar\phi$ depends on the major agent's strategy $\phi^0$, and the problem of the major agent effectively becomes a stochastic control problem.  We start with a reformulation of the definition of Stackelberg equilibrium in terms of $\bar \phi$ and $\phi
^0$ only. 
\begin{Lemma}\label{stacklem}
Let $(\bar\phi^*,\phi^{0*})$ be $\mathbb F^0$-adapted square integrable processes. There exists $\phi^*$ such that $(\phi^*,\bar\phi^*,\phi^{0*})$ is a Stackelberg mean field equilibrium if and only if the couple $(\bar\phi^*,\phi^{0*})$ satisfies the following conditions:
\begin{itemize}
    \item[i.] For every $\mathbb F^0$-adapted square integrable process $\nu$, 
$$
\mathbb E\left[\int_0^T \nu_t \{\alpha(t)\dot{\bar \phi}^*_t + S_t + a_0 \phi^{0*}_t + a \bar \phi^*_t\}dt + \lambda(\bar \phi^*_T - \overline X_T) \int_0^T \nu_t dt\right] = 0.
$$
\item[ii.] For every other couple $(\phi^{0},{\bar \phi})$ satisfying the condition i, the inequality \eqref{majopt} holds true. 
\end{itemize}
\end{Lemma}
\begin{proof}
Assume first that $(\phi^*,\bar\phi^*,\phi^{0*})$ is a Stackelberg mean field equilibrium. Then, For every $\mathbb F^0$-adapted square integrable process $\nu$,
$$
J^{MF}(\phi^* + \int_0^\cdot \nu_s ds, \overline \phi^*,\phi^{0*}) \leq J^{MF}(\phi^*, \overline \phi^*,\phi^{0*}). 
$$
Developing the functionals we get,
\begin{align*}
&\mathbb E\left[\frac12 \int_0^T \alpha(t)\nu_t^2 dt + \frac\lambda2 \left(\int_0^T \nu_t dt\right)^2\right]\\ & + \mathbb E\left[\int_0^T \nu_t\left\{\alpha(t)\dot \phi_t^* + S_t + a \bar \phi^*_t + a_0 \phi^{0*}_t \right\}dt + \lambda(\phi^*_T - X_T)\int_0^T \nu_t dt\right]\geq 0,
\end{align*}
and since $\nu$ is arbitrary, we see that this is equivalent to 
$$\mathbb E\left[\int_0^T \nu_t\left\{\alpha(t)\dot \phi_t^* + S_t + a \bar \phi^*_t + a_0 \phi^{0*}_t \right\}dt + \lambda(\phi^*_T - X_T)\int_0^T \nu_t dt\right]= 0.
$$
Taking conditional expectations and using Fubini's theorem, we then get condition i.~of the lemma. 

Assume now that conditions i.~and ii.~of the lemma hold true, and let $\phi^*$ be given by Proposition \ref{minor} applied to the couple $(\bar \phi^*, \phi^{0*})$. Define $\tilde \phi^*_t:= \mathbb E[\phi^*_t|\mathcal F^0_t]$. It remains to show that $\tilde \phi^*_t = \bar \phi^*_t$. Let $Y^*$ be an $\mathbb F^0$-martingale satisfying $Y^*_T = \lambda(\bar \phi^*_T - \overline X_T)$. By integration by parts, condition i. of the lemma is equivalent to 
$$
\mathbb E\left[\int_0^T \nu_t\left\{\alpha(t)\dot{\bar \phi}^*_t + S_t a_0 \phi^{0*}_t + a \bar \phi^*_t + Y^*_t\right\}dt\right] = 0,
$$
and since $\nu$ is arbitrary,
$$
\alpha(t)\dot{\bar \phi}^*_t + S_t+ a_0 \phi^{0*}_t + a \bar \phi^*_t + Y^*_t = 0,
$$
for all $t$. On the other hand, by Lemma \ref{minor}, taking the expectation with respect to $\mathbb F^0$, we get that there exists a $\mathbb F^0$-martingale $\widetilde Y$ with $\widetilde Y_T = \lambda(\tilde \phi_T - \overline X_T)$, and such that 
$$
\alpha(t)\dot{\tilde \phi}^*_t + S_t +a_0 \phi^{0*}_t + a \bar \phi^*_t + \widetilde Y_t = 0.
$$
Substracting this expression from the previous one, we get
$$
\alpha(t) (\dot{\bar \phi}^*_t - \dot{\tilde \phi}^*_t ) + Y^*_t - \widetilde Y_t = 0,\qquad Y^*_T - \widetilde Y_T = \lambda(\bar \phi^*_T - \tilde \phi^*_T)
$$
Thus, 
$$
{\bar \phi}^*_t - {\tilde \phi}^*_t = \int_0^t \frac{\widetilde Y_s - Y^*_s}{\alpha(s)}ds
$$
and therefore, using the terminal condition and the martingale property, 
$$
\widetilde Y_t - Y^*_t = \mathbb E[\widetilde Y_T - Y^*_T|\mathcal F^0_t] = \lambda \int_0^t\frac{\widetilde Y_s - Y^*_s}{\alpha(s)}ds + \lambda(\widetilde Y_t - Y^*_t)\int_t^T \frac{ds}{\alpha(s)}. 
$$
The unique solution of this linear equation is $\widetilde Y_t = Y^*_t$ for all $t$, and therefore $\tilde \phi^*_t = \bar \phi^*_t$ for all $t$.  
\end{proof}

The following proposition provides a martingale characterization of the Stackelberg mean field equilibrium. 
\begin{Proposition}\label{major}
Let $(\bar\phi^*,\phi^{0*})$ be $\mathbb F^0$-adapted square integrable processes. There exists $\phi^*$ such that $(\phi^*,\bar\phi^*,\phi^{0*})$ is a Stackelberg mean field equilibrium if and only if 
\begin{equation}\label{majorstrat}
    \dot \phi^{0*}_t =- \frac{M^0_t+S_t+a\overline{\phi}_t-a^0N_t}{\alpha_0(t)}, \quad \forall t\in [0,T],
\end{equation}
where $M^0$ is an $\mathbb{F}^0$-martingale and $N$ is an absolutely continuous $\mathbb{F}^0$-adapted process, and there exists an $\mathbb{F}^0$-martingale $M$, and an  $\mathbb{F}^0$-martingale $\overline Y$  such that the following system of equations is satisfied: 
\begin{equation} \label{adjointp}
    \left\{\begin{array}{lll}
    & M^0_T = a^0 N_T + a^0 \phi^0_T + \lambda_0(\phi^0_T - X^0_T) \\
    & M_t - a\phi^0_t + \alpha(t) \dot N_t - a N_t = 0,\qquad M_T = a \phi^0_T + (a+\lambda)N_T\\
    &\overline Y_t + \alpha(t) \dot{\bar \phi}_t + S_t + a \bar \phi_t + a^0\phi^0_t = 0,\qquad \overline Y_T = \lambda(\bar\phi_T - \overline X_T)
    \end{array}
    \right.
\end{equation}
\end{Proposition}
\begin{proof}
The optimization problem of the major agent consists in maximizing the objective function \eqref{majorobj} under the constraint of Lemma \ref{stacklem}, part i. Let us introduce the Lagrangian for this constrained optimization problem, which writes:
\begin{align*}
L(\phi^0, \bar \phi, \nu) &=   \mathbb E \left[\int_0^T  \frac{\alpha_0(t)}{2}\dot {\phi^0_t}^2 +(S_t+ {a}\overline \phi_t + {a}^0\phi^0_t)\dot \phi^0_t dt + \frac{\lambda_0}{2} (\phi^0_T-X^0_T)^2\right] \\ &+ \mathbb E\left[\int_0^T \nu_t \left\{\alpha(t)\dot{\bar \phi}_t + S_t + a^0 \phi^0_t + a \bar \phi_t\right\}dt + \lambda(\bar \phi_T - \overline X_T )\int_0^T \nu_t dt\right],
\end{align*}
where $\nu$ is a square integrable $\mathbb F^0$-adapted process. We claim that $\phi^0$ is the solution of the problem \eqref{majorobj} if and only if there exist $\nu$ and $\bar\phi$ such that $(\phi^0,\bar \phi)$ maximizes the Lagrangian $L(\cdot,\cdot,\nu)$, and $\bar \phi$ satisfies the constraint of Lemma \ref{stacklem}. Indeed, let $(\phi^0,\bar \phi, \nu)$ be such a triple and $(\phi^{0\prime}, \bar \phi')$ be another pair of strategies satisfying the constraint of Lemma \ref{stacklem}. Then, 
$$
L(\phi^0, \bar \phi, \nu) \geq L(\phi^{0\prime}, \bar \phi', \nu),
$$
and since both $\bar \phi$ and $\bar \phi'$ satisfy the constraint of Lemma \ref{stacklem}, this implies that inequality \eqref{majopt} holds true.

We now turn to the problem of maximizing the Lagrangian. Let $N_t = \int_0^t \nu_s ds$. The first order condition for $\phi_0$ writes: there exists a martingale $M^0$ such that
$$
M^0_t  +  \alpha_0(t) \dot \phi^0_t + S_t + a\bar \phi_t - a^0 N_t = 0,\qquad M^0_T = a^0 N_T + a^0 \phi^0_T + \lambda_0(\phi^0_T - X^0_T).  
$$
The first order condition for $\bar \phi$ writes: there exists a martingale $M$ such that
$$
M_t - a\phi^0_t + \alpha(t) \dot N_t - a N_t = 0,\qquad M_T = a \phi^0_T + (a+\lambda)N_T. 
$$
Finally, the last condition is given by the constraint that $\bar \phi$ is optimal for the generic minor agent. Hence conditioning \eqref{minoreq} by the common noise, there exists a martingale $\overline Y$ such that
$$
\overline Y_t + \alpha(t) \dot{\bar \phi}_t + S_t + a \bar \phi_t + a^0\phi^0_t = 0,\qquad \overline Y_T = \lambda(\bar\phi_T - \overline X_T). 
$$
\end{proof}
The following theorem provides an explicit characterization of the equilibrium in the Stackelberg setting.

\begin{Theorem}[Explicit solution]\label{stackexplicit}
Let  $\Xi_t = (\phi^0_t, N_t, \bar \phi_t)'$. The unique equilibrium of the mean field game with a major agent is characterized by the following differential equation: 
\begin{align}
B(t)^{-1}A\, \Xi_t + \dot \Xi_t = -\left(\begin{aligned} &\alpha_0(t)^{-1}(M^0_t +S_t) \\  &\alpha(t)^{-1}M_t \\  &\alpha(t)^{-1}(\overline Y_t + S_t)\end{aligned}\right),\label{inhom}
\end{align}
where $N$ is a $\mathbb{F}^0$-adapted process with $N_0=0$ and $M^0$, $M$, $\overline Y$ are $\mathbb{F}^0$-martingales that satisfy: 
\begin{equation}\label{system}
    \left\{\begin{array}{lll}
     & M^0_T = a^0 N_T + a^0 \phi^0_T + \lambda_0(\phi^0_T - X^0_T) \\
   & M_T = a \phi^0_T + (a+\lambda)N_T\\
    & \overline Y_T = \lambda(\bar\phi_T - \overline X_T)
    \end{array}
    \right.
\end{equation} 
and $$
A = \left(\begin{aligned}
    && 0 && -a^0 && a \\
    && -a && -a && 0 \\
    && a^0 && 0 && a
  \end{aligned}\right),\qquad B(t) = \left(\begin{aligned}
    && \alpha_0(t) && 0 && 0 \\
    && 0 && \alpha(t) && 0 \\
    && 0 && 0 && \alpha(t)
  \end{aligned}\right).
$$
Denoting by $\Phi(t)$ the fundamental matrix solution of the equation $B(t)^{-1}A\, \Xi_t + \dot \Xi_t  = 0$, the solution is given in the integral form by the following expression: 
\begin{align*}
 \Xi_t &= \Upsilon_t -\Pi_{0,t} (I + D\Pi_{0,T})^{-1} (D\widetilde \Upsilon_0 - \Lambda \mathcal X_0) \\&- \int_0^t \Pi_{s,t} (I + D\Pi_{s,T})^{-1} (Dd\widetilde \Upsilon_s - \Lambda d\mathcal X_s). 
\end{align*}
where
\begin{align*}
\Upsilon_t &:=  - \Phi(t)\int_0^t \Phi(s)^{-1} \left(\begin{aligned}&\alpha_0(s)^{-1}S_s \\  &0 \\  &\alpha(s)^{-1}S_s\end{aligned}\right)ds,\quad \mathcal X_s:=\left(\begin{aligned}&X^0_s  \\  &0 \\  &\overline X_s\end{aligned}\right),\quad\\ D &:= \left(\begin{aligned}
    &&a^0 + \lambda_0 &&  a^0 && 0\\
    && a && a+\lambda && 0\\
    && 0 && 0 && \lambda
  \end{aligned}\right),\quad \Lambda := \left(\begin{aligned}
    && \lambda_0 &&  0 && 0\\
    && 0 &&0&& 0\\
    && 0 && 0 && \lambda
  \end{aligned}\right),\\
  \Pi_{s,t}&:=\Phi(t)\int_s^t \Phi(u)^{-1} B(u)^{-1} du \quad  \text{and} \quad \widetilde \Upsilon_t = \mathbb E[\Upsilon_T|\mathcal F_t].
  \end{align*}           
\end{Theorem}
\begin{Remark}\label{remstack}
If $\alpha_0(t) = c\alpha(t)$ for some constant $c$, the fundamental matrix solution is given explicitly by
$$
\Phi(t) = \exp\left(-\int_0^t B(s)^{-1}A ds\right)
$$
\end{Remark}

\begin{proof}
From equations \eqref{minoreq} in Proposition \ref{minor} and \eqref{adjointp} in Proposition \ref{major}, we immediately deduce the expression of the characterizing differential equation of the equilibrium \eqref{inhom}.

Let $\Phi(t)$ be the fundamental matrix solution of the equation $B(t)^{-1}A\, \Xi_t + \dot \Xi_t = 0 $, that is, for every $C\in \mathbb R^3$, the solution with initial condition $\Xi_0 = C$ is given by $\Phi(t)C$. By variation of constants we have that the solution of \eqref{inhom} is given by:
$$
\Xi_t =  - \Phi(t)\int_0^t \Phi(s)^{-1} \left(\begin{aligned}&\alpha_0(s)^{-1}(M^0_s +S_s) \\  &\alpha(s)^{-1}M_s \\  &\alpha(s)^{-1}(\overline Y_s + S_s)\end{aligned}\right)ds. 
$$
Letting:
$$ \mathcal M_s:=\left(\begin{aligned}M^0_s  \\  M_s \\  \overline Y_s\end{aligned}\right) \quad \text{and}\quad \widehat \Xi_t = \Xi_t - \Upsilon_t,
$$
we obtain the simplified equation:
$$
\widehat\Xi_t =  - \Phi(t)\int_0^t \Phi(s)^{-1} B(s)^{-1}\mathcal M_s ds. 
$$
and finally, using \eqref{system} and the martingale property, the martingale components satisfy:
\begin{align*}
\mathcal M_t =-D \Phi(T)\int_0^t \Phi(s)^{-1} B(s)^{-1}\mathcal M_s ds -D \Pi_{t,T}\mathcal M_t + D\widetilde \Upsilon_t -\Lambda \mathcal X_t,
\end{align*}
From this we deduce, on the one hand,
\begin{align*}
\mathcal M_0 &= (I + D\Pi_{0,T})^{-1} (D\widetilde \Upsilon_0 -\Lambda\mathcal X_0),
\end{align*}
and on the other hand,
\begin{align*}
(I+D\Pi_{t,T})d\mathcal M_t = Dd\widetilde \Upsilon_t -\Lambda d\mathcal X_t,
\end{align*}
so that finally:

\begin{align*}
 \Xi_t &= \Upsilon_t +  \int_0^t d\Pi_{s,t}\cdot \mathcal M_s = \Upsilon_t- \Pi_{0,t}\mathcal M_0 - \int_0^t \Pi_{s,t} d\mathcal M_s \\ &= \Upsilon_t -\Pi_{0,t} (I + D\Pi_{0,T})^{-1} (D\widetilde \Upsilon_0 - \Lambda \mathcal X_0) \\ &- \int_0^t \Pi_{s,t} (I + D\Pi_{s,T})^{-1} (Dd\widetilde \Upsilon_s - \Lambda d\mathcal X_s). 
\end{align*}
\end{proof}

Let us make some comments on how minor and major player strategies change when the parameters of the model vary. First, when the major player has no price impact, $a_0 = 0$, we recover the homogeneous mean field setting optimal strategy for the minor player from \eqref{adjointp} and \eqref{system}:
$$\bar \phi_t = \int_0^t\frac{\eta(s,t)}{\alpha(s)}(\overline{Y}_s+S_s)ds, $$
where 
$\overline{Y}$ satisfies the equation:
\begin{align*}\left\{
    \begin{array}{ll}
       &\overline{Y}_t + \alpha(t) \dot{\bar{\phi}}^*_t + S_t + a \bar \phi^*_t = 0  \\
       &\overline{Y}_T = - \lambda(\bar \phi^*_T - \overline{X}_T).
    \end{array}\right.
\end{align*}
Second, we explore the limiting behavior of the optimal strategies for the major agent and for the mean field in various limiting cases. In this corollary, we use the notation of Theorem \ref{stackexplicit}.
\begin{Corollary}\label{limit2.cor}
${}$
\begin{itemize}
\item[i.] Assume that the fundamental price process $S$ is a martingale. Then the equilibrium mean field position of minor agents, and the position of the major agent satisfy
\begin{align*}
 \Xi_t &=  -\Pi_{0,t} (I + D\Pi_{0,T})^{-1} \left(\begin{aligned}S_0 - \lambda_0 X^0_0 \\ 0 \\ S_0 - \lambda \overline X_0\end{aligned}\right) \\ &- \int_0^t \Pi_{s,t} (I + D\Pi_{s,T})^{-1} d\left(\begin{aligned}  S_s - \lambda_0 X^0_s \\  0 \\ S_s - \lambda \overline X_s\end{aligned}\right). 
\end{align*}

\item[ii.] In the limit of infinite terminal penalty (when $\lambda,\lambda_0\to \infty$)  the equilibrium mean field position of minor agents, and the position of the major agent satisfy,
$$
\Xi_t \to \Upsilon_t -\Pi_{0,t} \Pi_{0,T}^{-1} (\widetilde \Upsilon_0 - D_\infty \mathcal X_0) - \int_0^t \Pi_{s,t} \Pi_{s,T}^{-1} (d\widetilde \Upsilon_s - D_\infty d\mathcal X_s),
$$
almost surely for all $t\in [0,T]$, where
$$
D_\infty = \left(\begin{aligned}
    && 1 &&  0 && 0\\
    && 0 && 0 && 0\\
    && 0 && 0 && 1
  \end{aligned}\right).
$$
When the fundamental price process $S$ is a martingale, in the limit of infinite terminal penalty, the strategies do not depend on the fundamental price and we have,
$$
 \Xi_t \to   \Pi_{0,t}\, \Pi_{0,T}^{-1}\, \left(\begin{aligned}  X^0_0 \\ 0 \\  \overline X_0\end{aligned}\right) +\int_0^t \Pi_{s,t}\, \Pi_{s,T}^{-1}\, d\left(\begin{aligned} X^0_s \\ 0 \\ \overline X_s\end{aligned}\right). 
$$

\item[iii.] In the absence of terminal penalties (when $\lambda=\lambda_0=0$), the equilibrium mean field position of minor agents, and the position of the major agent satisfy,
\begin{align}
 \Xi_t &= \Upsilon_t -\Pi_{0,t} (I + D_0\Pi_{0,T})^{-1} D_0\widetilde \Upsilon_0  - \int_0^t \Pi_{s,t} (I + D_0\Pi_{s,T})^{-1} D_0 d\widetilde \Upsilon_s . \label{myopicstack}
\end{align}

\end{itemize}
\end{Corollary}

\begin{proof}
The first part is a simplification of the proof of Theorem \ref{stackexplicit}. For the second part, we can rewrite:
\begin{align*}
 \Xi_t &= \Upsilon_t -\Pi_{0,t} (I + D\Pi_{0,T})^{-1} (D\widetilde \Upsilon_0 - \Lambda \mathcal X_0) \\ &- \int_0^t \Pi_{s,t} (I + D\Pi_{s,T})^{-1} (Dd\widetilde \Upsilon_s - \Lambda d\mathcal X_s)\\
 & = \Upsilon_t -\Pi_{0,t} (D^{-1} + \Pi_{0,T})^{-1} (\widetilde \Upsilon_0 - D^{-1}\Lambda \mathcal X_0) \\ &- \int_0^t \Pi_{s,t} (D^{-1} +\Pi_{s,T})^{-1} (d\widetilde \Upsilon_s - D^{-1}\Lambda d\mathcal X_s) 
\end{align*} 
and when $\lambda, \lambda_0 \longrightarrow \infty$, $D^{-1}\to 0$ and $D^{-1}\Lambda \to D_\infty$. The third part follows by direct substitution of $\lambda = \lambda_0 = 0$ into the general formula. 
\end{proof}

Interestingly, when the players don't have a terminal penalty ($\lambda = \lambda_0 = 0$), the equilibrium positions of the agents in equation \eqref{myopicstack} still contain forward looking terms, which were absent in the case of the mean field game with identical players (see Equation \eqref{strathomo} with $\lambda=0$). The presence of these terms is due to the strategic interaction of the major player with the mean field of small agents. 

In the limit of zero trading costs, the gain of the major player remains bounded in expectation, however, contrary to the case of identical players, the optimal strategy of the major agent cannot be determined uniquely from the optimization problem. Indeed, assuming that the trading cost for minor agents is zero, the equilibrium price (computed from Equation (12) in \cite{feron2020price} with $N\to \infty$) is given by
$$
P_t = S_t + a_0 \phi^0_t + a \bar \phi_t = \frac{\lambda}{a+\lambda}(a\overline X_t + \mathbb E[S_T|\mathcal F_t] + a_0 \mathbb E[\phi^0_T |\mathcal F_t]). 
$$
Substituting this expression into the optimization problem for the major player, we need to minimize the following functional:
\begin{align*}
&\mathbb E\left[\frac{\lambda}{a+\lambda}\int_0^T \dot\phi^0_t(a\overline X_t + \mathbb E[S_T|\mathcal F_t] + a_0 \mathbb E[\phi^0_T |\mathcal F_t])dt + \frac{\lambda_0}{2}(\phi^0_T - X^0_T)^2 \right]\\
& = \mathbb E\left[\frac{\lambda}{a+\lambda}(a\phi^0_T \overline X_T + \phi^0_T S_T + a_0(\phi^0_T)^2 )+  \frac{\lambda_0}{2}(\phi^0_T - X^0_T)^2 \right],
\end{align*}
where the equality follows, in particular, from the martingale property of $\overline X_T$. Since the expression to be minimized only depends on the terminal value $\phi^0_T$ of the major agent's position, any strategy with the optimal terminal value will satisfy the condition of optimality: the Stackelberg equilibrium will not be unique in this case.

 To finish this section, we provide the explicit form of the strategy of the minor agents.
\begin{Corollary}[Minor agent strategy]\label{minorbis}
Under assumption \ref{stackass}, the optimal generic minor agent position $\phi^*$ is given by: 
$$  \phi^*_t = \int_0^t \widetilde \Delta_{s,t} \frac{\lambda d\check  X_s}{1+\lambda\widetilde\Delta_{s,T}} + \widetilde \Delta_{0,t}\frac{\lambda \check X_0}{1+\lambda\widetilde\Delta_{0,T}}  + \bar{\phi}^*_t$$
where $\bar \phi^*$ is the optimal aggregate position of the minor agents, as given by Theorem \ref{stackexplicit}. 
\end{Corollary}

\begin{proof}

Let $\check\phi^*_t = \phi^{*}_t - \bar \phi^*_t$, $\check X_t = X_t - \overline X_t$ and $\check Y_t:= Y_t -
\overline Y_t$. Then, from the explicit form of $Y$ and $\overline Y$ in Proposition \ref{minor} it follows that $\check Y$ is an $\mathbb F$-martingale
and satisfies
$$
\check Y_T = -\lambda(\check \phi^{*}_T -\check X_T),\qquad
\check Y_t = \alpha(t) \dot{\check \phi}^{*}_t.
$$
Then,
\begin{align}
\check\phi^*_t = \int_0^t \frac{\check Y_s}{\alpha(s)} ds,\label{phiY}
\end{align}
and by the martingale property,
$$
Y_t = -\lambda \mathbb E[\check \phi^{*}_T -\check X_T|\mathcal F_t] = -\lambda \int_0^t \frac{\check Y_s}{\alpha(s)} ds - \lambda \check Y_t \int_t^T \frac{ds}{\alpha(s)} + \lambda \check X_t. 
$$
Solving this linear equation for $\check Y$ then substituting into \eqref{phiY}, we get the result. 
\end{proof}

\section{Approximate Nash equilibrium in the $N$-player Stackelberg game}\label{epsnashstackelberg}

In this section, we derive the $\epsilon$-Nash approximation for the Stackelberg game. In the present leader-follower setting, we allow the minor agents to change their strategies when the major agent deviates from her optimal one.

{

Since we would like to study the rate of convergence as $N\to \infty$, we assume that there is a major player and an infinity of minor players replacing the generic agent. Their demand forecasts are respectively given by $X_t^i$, $i=0,\dots,\infty$, $t \in [0,T]$. The private demand forecasts of all agents are defined on the same probability space. We therefore impose the following assumption.
\begin{Assumption}${}$\label{samespace.ass}
\begin{itemize}
\item The process $S$ is square integrable and adapted to the filtration $\mathbb F^0$. 
\item The demand forecast $X^0$ of the major agent is a square integrable $\mathbb F^0$-martingale. 
\item The processes $(X^i)_{i=1}^\infty$ are square integrable $\mathbb F$-martingales.
\item There exists a square intergrable $\mathbb F$-martingale $\overline X$, such that for all $i\geq 1$, and all $t\in[0,T]$, almost surely, $\mathbb E[X^i_t|\mathcal F^0_t]=\overline X_t$.
\item The processes $(\check X^i)_{i=1}^\infty$ defined by $\check X^i_t = X^i_t - \overline X_t$ for $t\in[0,T]$, are orthogonal square integrable $\mathbb F$-martingales, such that the expectation $\mathbb E[(\check X^i_T)^2]$ does not depend on $i$. 
\end{itemize}
\end{Assumption}

The strategy $(\dot \phi^i)$ of agent $i=1,\dots,\infty$ is said to be admissible if it is $\mathbb F$-adapted and square integrable; the strategy $(\dot\phi^0)$ of the major agent is admissible if it is $\mathbb F^0$-adapted and square integrable. 
For a fixed $N\geq 1$, we denote: 
\begin{align*}
    &P^N(\phi^0_t,...,\phi^N_t) = S_t + {a} \overline \phi^N_t + {a}^0 \phi_t^0\\ & P^{MF}(\phi^0_t,\bar \phi_t) = S_t+{a} \bar \phi_t+ {a}^0 \phi_t^0,
\end{align*}
where $\bar \phi^N_t = \frac{1}{N}\sum_{i = 1}^N \phi^i_t$ is the average position of the minor agents.
And we define in the $N$-player game, the objective functions for {the major agent:
\begin{multline}\
       J^{N,0}(\phi^0, \phi^{-0}) \\:= - \mathbb{E}\left[\int_{0}^{T}\left\{\frac{\alpha_0(t)}{2}(\Dot{\phi^0_t})^{2}+\Dot{\phi^0_{t}}P^N(\phi^0_t,\dots,\phi^N_t) \right\}dt+\frac{\lambda_0}{2}(\phi^0_{T}- X^0_{T})^2\right], 
\end{multline}
and for the minor agents $i=1,\dots,N$}:
\begin{equation}\
       J^{N,i}(\phi^i, \phi^{-i}) := - \mathbb{E}\left[\int_{0}^{T}\left\{\frac{\alpha(t)}{2}(\Dot{\phi^i_t})^{2}+\Dot{\phi^i_{t}}P^N(\phi^0_t,\dots,\phi^N_t) \right\}dt+\frac{\lambda}{2}(\phi^i_{T}- X^i_{T})^2\right],
\end{equation}
{as well as the objective function for the minor agents $i = 1, \dots, N$, in the mean field setting:}
\begin{equation}
       J^{MF}(\phi^i, {\bar\phi}, \phi^0) := - \mathbb{E}\left[\int_{0}^{T}\left\{\frac{\alpha(t)}{2}(\Dot{\phi^i_t})^{2}+\Dot{\phi^i_{t}}P^{MF}(\phi^0_t,\bar\phi_t) \right\}dt+\frac{\lambda}{2}(\phi^i_{T}- X^i_{T})^2\right].
\end{equation}

{We next provide a definition of the $\epsilon$-Nash equilibrium in the present Stackelberg setting. As mentioned above, the deviations of the major and minor agents must be treated differently: when the major agent deviates, we allow the minor agents to adjust their strategies to respond optimally to the new strategy of the major agent. We say that the minor agent strategies $\phi^1,\dots,\phi^N$ are an optimal response to the major agent strategy $\phi^0$ if for every $i=1,\dots,N$ and for every admissible minor agent strategy $\tilde\phi^i$,
$$
J^{N,i}(\tilde \phi^i,\phi^{-i}) \leq J^{N,i}(\phi^i,\phi^{-i}). 
$$
\begin{Definition}[Stackelberg $\varepsilon$-Nash equilibrium]\label{stackepsnash}
We say that $(\phi^{i*}_t)_{t \in [0,T],0 \leq i\leq N}$ is an $\epsilon$-Nash equilibrium for the N-player game if these strategies are admissible and the following holds.
\begin{itemize}
    \item[i.] \textbf{Deviation of a minor player:}\quad  For any other admissible strategy $\phi^{i}$ for the minor player $i$, $i=1,\dots,N$, 
    $$
    J^{N,i}(\phi^{i},\phi^{-i*})-\varepsilon\leq J^{N,i}( \phi^{i*},\phi^{-i*}).
    $$
    \item[ii.] \textbf{Deviation of the major player:}\quad For any other set of admissible strategies $(\phi^i),i=0,\dots,N$, such that $\phi^1,\dots,\phi^N$ are optimal responses of minor players to the major player strategy $\phi^0$, we have,
\begin{align*}
    J^{N,0}( \phi^0, \phi^{-0})-\varepsilon\leq J^{N,0}( \phi^{0*}, \phi^{-0*}).
    \end{align*}
\end{itemize}
\end{Definition}}

Our definition of $\varepsilon$-Nash equilibrium is different from the one in \cite{carmona2016probabilistic}: while the latter paper assumes that the major player deviates from her strategy unilaterally (see Definition 4.2), we allow the minor players to respond to the deviation of the major player, in agreement with the leader-follower nature of the game. In addition, in \cite{carmona2016probabilistic}, an a priori bound on the $L
^p$-norm of the new strategy of the major agent is required to establish Theorem 4.1, whereas no such bound is needed in our setting. 

\begin{Proposition}\label{epsnash} 
Assume that the strategies of the $N$ minor agents are given by
$$  
\phi^{i*}_t = \int_0^t \widetilde \Delta_{s,t} \frac{\lambda d\check  X^i_s}{1+\lambda\widetilde\Delta_{s,T}} + \widetilde \Delta_{0,t}\frac{\lambda \check X^i_0}{1+\lambda\widetilde\Delta_{0,T}}  + \overline{\phi}^*_t,$$
where $\bar \phi$ is the third component of the mean field equilibrium defined in Theorem \ref{stackexplicit}. Assume that the strategy of the major agent is given by Theorem \ref{stackexplicit} as well. 
Let assumption \ref{samespace.ass} hold true. Then there exists a constant $C<\infty$, which does not depend on $N$, such that 
 these strategies form an $\varepsilon$-Nash equilibrium of the N-player game with $\varepsilon = \frac{C}{N^{1/2}}$.  
\end{Proposition}

\begin{Remark}{
The $\varepsilon$-Nash equilibrium described in Proposition \ref{epsnash} approximates the $N$-player equilibrium in the complete information setting (where every player observes the others' actions), but its implementation for each agent only requires the knowledge of the common information $\mathbb F^0$ as well as the agent's individual forecast.} 
\end{Remark}

\begin{proof}
We need to show conditions i. and ii. of Definition \ref{stackepsnash}. Condition i.~is shown in the same way as in the case of homogeneous players (see proof of Proposition 2 in \cite{feron2020price}). We therefore focus on condition ii.
Assume that all agents change their strategies to new ones $\phi^{0},\dots, \phi^{N}$, such that $\phi^1,\dots,\phi^N$ are optimal responses to $\phi^0$.  Let $\bar \phi$ be the optimal "mean field" response to the major agent strategy $\phi^0$. 

\paragraph{\textbf{Step 1.}} We first suppose that there exists a constant $ A>0$, such that  $\mathbb{E}\left[\int_0^T (\dot \phi_t^0)^2dt\right]< A$. 

By Proposition 1 in \cite{feron2020price}, for some constants $c$ and $C$ which do not depend on $N$, and may change from line to line,
$$\mathbb{E}\left[ \int_0^T (\bar \phi_t-\bar \phi^N_t)^2dt\right]  \leq \frac{C}{N^2} \mathbb{E}\left[\int_0^T (S_s+ a^0\phi^0_s)^2ds\right]+ \frac{c}{N}, $$
and by our assumption,  
\begin{align*}
&\mathbb{E}\left[\int_0^T (S_s+ a^0\phi^0_s)^2ds\right] \\ &= \mathbb{E}\left[\int_0^T S^2_sds\right] + 2a^0\mathbb{E}\left[\int_0^T S^2_sds\right]^{\frac{1}{2}}\mathbb{E}\left[\int_0^T (\phi^0_s)^2ds\right]^{\frac{1}{2}} + \mathbb{E}\left[\int_0^T (a^0\phi^0_s)^2ds\right]\\ &< C(1+A).
\end{align*}
Thus using Cauchy-Schwartz inequality,
\begin{align*}
&J^{N,0}(\phi^{0},\phi^{-0}) - J^{N,0}(\phi^{0*},\phi^{-0*}) \\
&=  J^{N,0}(\phi^{0},\phi^{-0}) - J^{MF}(\phi^{0},\bar\phi)+J^{MF}(\phi^{0},\bar\phi) - J^{MF}(\phi^{0*},\bar\phi^*) \\&\qquad\qquad+ J^{MF}(\phi^{0^*},\bar\phi^*) - J^{N,0}(\phi^{0*},\phi^{-0*}) \\
&\leq J^{N,0}(\phi^{0},\phi^{-0}) - J^{MF}(\phi^{0},\bar\phi^0) + J^{MF}(\phi^{0^*},\bar\phi^{*}) - J^{N,0}(\phi^{0*},\phi^{-0*})\\
& \leq a \Big\{ \mathbb E\left[\int_0^T (\dot\phi^{0}_t)^2 dt\right]^{\frac{1}{2}}\mathbb E\left[\int_0^T (\bar \phi_t - \bar \phi^{N}_t)^2 dt\right]^{\frac{1}{2}} \\ &\qquad\qquad+ \mathbb E\left[\int_0^T (\dot\phi^{0*}_t)^2 dt\right]^{\frac{1}{2}}\mathbb E\left[\int_0^T (\bar \phi^*_t - \bar \phi^{N*}_t)^2 dt\right]^{\frac{1}{2}}\Big\}\\
& = \mathcal{O}\left(N^{-\frac{1}{2}}\right)
\end{align*}

\paragraph{\textbf{Step 2.}}
{Letting $\bar \alpha_0 = \min_{0\leq t\leq T} \alpha_0(t)$, and $\bar \alpha = \min_{0\leq t\leq T} \alpha(t)$, we have, by definition, 
\begin{align}
J^{N,0}(\phi^{0},\phi^{-0}) &\leq - \mathbb E\left[\int_0^T\frac{\bar  \alpha_0}{2}(\dot\phi^0_t)^2 + \dot\phi^0_t (S_t + a \bar \phi^{N}_t +a_0\phi^{0}_t )dt + \frac{\lambda_0}{2}(\phi^0_T - X_T)^2 \right]\notag 
\end{align}
On the other hand, since $\phi^{i}$ for $i=1,\dots,N$ are optimal responses to $\phi^0$, we get that
\begin{multline*}
\mathbb E\left[\int_0^T\left\{\frac{\alpha(t)}{2} (\dot\phi^i_t)^2 + \dot\phi^i_t (S_t + a \bar \phi^{N}_t +a_0\phi^{0}_t )\right\}dt + \frac{\lambda}{2}(\phi^i_T - X^i_T)^2 \right] \\ \leq \frac{\lambda}{2} \mathbb E[(X^i_T)^2] \leq C,
\end{multline*}
where $C<\infty$ is defined by $C:=\max_i \frac{\lambda}{2} \mathbb E[(X^i_T)^2]$. Summing up the above inequality over $i=1,\dots,N$, dividing by $N$ and using Jensen's inequality, we get 
$$
\mathbb E\left[\int_0^T\left\{\frac{\bar\alpha}{2} (\dot{\bar\phi}^N_t)^2 + \dot{\bar\phi}^N_t (S_t + a \bar \phi^{N}_t +a_0\phi^{0}_t )\right\}dt + \frac{\lambda}{2}(\bar\phi^N_T - \overline X^N_T)^2 \right] \leq C.
$$
Multiplying this inequality by $\frac{a}{a_0}$, adding it to the first one, and using integration by parts, we finally get
\begin{align*}
J^{N,0}(\phi^{0},\phi^{-0}) &\leq C - \mathbb E\Bigg[\int_0^T\left\{\frac{\bar \alpha_0}{2} (\dot\phi^0_t)^2 + \frac{a\bar\alpha}{2a_0} (\dot{\bar\phi}^N_t)^2+ \frac{S_t}{a_0}(a_0 \dot\phi^0_t + a \dot{\bar\phi}^N_t)\right\}dt \\ & + \frac{1}{2a_0}(a \bar \phi^{N}_T +a_0\phi^{0}_T )^2 + \frac{\lambda_0}{2}(\phi^0_T - X_T)^2+\frac{a\lambda}{2a_0}(\bar\phi^N_T - \overline X^N_T)^2 \Bigg]\notag \\
&\leq C - \mathbb E\left[\int_0^T\frac{\bar \alpha_0}{2} (\dot\phi^0_t)^2 dt \right] +\mathbb E\left[\int_0^T S_t^2 dt\right]^{\frac{1}{2}} \mathbb E\left[\int_0^T (\dot\phi^0_t)^2 dt\right]^{\frac{1}{2}}\\ &- \mathbb E\left[ \int_0^T \frac{a\bar\alpha}{2a_0} (\dot{\bar\phi}^N_t)^2 dt\right]+ \frac{a}{a_0}\mathbb E\left[\int_0^T S_t^2 dt\right]^{\frac{1}{2}} \mathbb E\left[\int_0^T (\dot{\bar\phi}^N_t)^2 dt\right]^{\frac{1}{2}}\\
&\leq C - \mathbb E\left[\int_0^T\frac{\bar \alpha_0}{2} (\dot\phi^0_t)^2 dt \right] +\mathbb E\left[\int_0^T S_t^2 dt\right]^{\frac{1}{2}} \mathbb E\left[\int_0^T (\dot\phi^0_t)^2 dt\right]^{\frac{1}{2}}\\ &\qquad\qquad+  \frac{a}{\bar \alpha a_0}\mathbb E\left[\int_0^T S_t^2 dt\right]
\end{align*}
Thus, if 
$$
\mathbb E\left[\int_0^T (\dot\phi^0_t)^2dt\right] > A, 
$$
for $A$ sufficiently large (but not depending on $N$), then from the above estimate it follows that 
$$
J^{N,0}(\phi^{0},\phi^{-0}) \leq J^{N,0}(\phi^{0*},\phi^{-0*}), 
$$
as well. 
}

\end{proof}}

 \section{Numerical illustration}\label{empirical}

In this section, our objective is to illustrate the theoretical results presented in sections \ref{extension} and \ref{epsnashstackelberg} with numerical simulations. We analyze the role of the major producer in the market and its impact on price characteristics such as volatility and price-forecast correlation, and compare this situation to the homogeneous agent setting studied in \cite{feron2020price}. Some comparisons with the empirical market characteristics are also performed, but we refer the reader to \cite{feron2020price} and other papers cited therein for a more detailed description of intraday electricity markets and their empirical features. As empirical analysis in \cite{feron2020price} are led on actual wind infeed forecasts, we will consider production forecasts here instead of demand forecasts as it is the case in the rest of the paper. Throughout, we consider that the production forecasts are the differences between actual production forecasts and the agents' positions in the market at time 0. Therefore, the initial values $X^i_0$, $i=0, \dots, N$ will be set to 0.

\paragraph{\textbf{Model specification}}
We now define the dynamics for the fundamental price and for the production forecasts used in the simulations and specify the parameter values. The objective being to illustrate the model, the majority of the parameters are not precisely estimated, but are given ad hoc plausible values.

The evolution of the fundamental price is described as follows: 
\begin{equation}
    dS_t = \sigma^S dW_t
\end{equation}
where $\sigma^s$ is a constant and $(W_t)_{t \in [0,T]}$ is Brownian motion. 
We also assume that the liquidity functions $\alpha(.)$ and $\alpha_0(.)$ have a specific form given by
\begin{eqnarray}
\label{alpha}
    \alpha(t) & = & \alpha \times (T-t)+ \beta, \quad \forall t \in [0, T] \\
    \alpha_0(t) & = & \alpha_0 \times (T-t)+ \beta_0, \quad \forall t \in [0, T]
\end{eqnarray}
where $\alpha, \beta$, $\alpha_0$ and $\beta_0$ are strictly positive constants.  The liquidity functions are thus decreasing with time. This assumption relies on the fact that the market becomes more liquid as we get closer to the delivery time and it is less costly to trade when the market is liquid.

To simulate production forecasts we assume the following dynamics: 
\begin{eqnarray}
 d \bar X_{t} & = & \bar\sigma d\bar B_t \\
 d \check X^0_{t} & = & \sigma^0 dB^0_{t},\\
 d \check X^i_{t} & = & \sigma^X dB^i_{t}, \qquad i \in \{1, \dots, N\}
\end{eqnarray} 
where $\bar \sigma$, $\sigma^0$ and $\sigma^X$ are  constants and $(\bar B_t)_{t \in [0,T]}$, $(B^i_{t})_{t \in [0,T]}$,  $i \in \{0, \dots, N\}$ are independent Brownian motions, also independent from $(W_t)_{t \in [0,T]}$.

In this illustration, we choose the same parameters for the dynamics of the common and the individual production forecasts, as well as the forecast of the major agent. The common volatility is calibrated to wind energy forecasts in Germany over January 2015 during the last quotation hour, by using the classical volatility estimator 
\begin{equation}\label{volforecast}
   \bar \sigma = \sigma^0 = \sigma^X = \hat{\sigma} = \frac{\sqrt{\Delta t}}{n' - 1}\sum_{i=1}^{n'}Y_i^2
 \end{equation}
with $\Delta t$ the time step between two observations, $Y_i = X_{t_i} - X_{t_{i-1}}$ the increment between two successive observations and $n'$ the total number of observed increments. 
As the forecasts are updated every 15 minutes, there are three daily variations during the last hour of forecasts from the 3$^{rd}$ of January to the 31$^{th}$ of January. Thus, for each delivery hour we dispose of $n'= 87$ increments to estimate the volatility.

The  model parameters are specified in Table \ref{parameters}.

\begin{table}[h]
\centering
   \begin{tabular}{ |c|c||c|c|}
     \hline
     Parameter & Value & Parameter & Value \\ \hline
     $S_0$ & $40$ \euro{}/MWh &  $a$ & $1$  \euro{}/MWh$^{2}$  \\ \hline
     $\sigma^{S}$ & $10$ \euro{}/MWh.h$^{1/2}$ & $\lambda$ & $100$ \euro{}/MWh$^{2}$\\ \hline
     $\overline{X}_0$ & 0 MWh & $\lambda_0$ & $100$ \euro{}/MWh$^{2}$\\ \hline
      $\bar \sigma$ &  $73$ MWh/h$^{1/2}$ &  $\alpha$ & 0.3 \euro{}/h.MW$^{2}$ \\ \hline
     $\check X^{i}_0$ & 0 MWh & $\alpha_0$ &  0.3 \euro{}/h.MW$^{2}$ \\ \hline
       $\sigma^{X},\sigma^{0}$ & $73$ MWh/h$^{1/2}$  & $\beta$ & 0.1 \euro{}/MW$^{2}$  \\ \hline
     $N$ & 100 & $\beta_0$ & 0.1 \euro{}/MW$^{2}$  \\ \hline

   \end{tabular}
   \caption{Parameters of the model}
   \label{parameters}
 \end{table}

\paragraph{\textbf{Equilibrium price and market impact}}

In Figure \ref{Fig5}, we plot the major agent production forecast and the common production forecast (respectively the orange and blue solid lines) together with the equilibrium position of the major agent and the aggregate position of the minor agents given by Proposition \ref{major} and Proposition \ref{minorbis} (respectively the orange and blue dashed lines). For comparison, we also plot the aggregate position in the identical agent case (dotted green line). All trajectories have been computed with the same production forecasts, the same fundamental price, initial values, volatilities and parameters, specified in Table \ref{parameters}, except for the price impact coefficients of major and minor player, which differ according to model specification. In the Stackelberg game we chose ${a}^0 = {a} = 0.5$ \euro{}/MWh$^{2}$, and in the homogeneous case we kept ${a} = 1$ \euro{}/MWh$^{2}$ .

\begin{figure}[h]
      \centering
    \caption{Stackelberg game}
    \label{Fig5}
   \includegraphics[width = 0.60\textwidth]{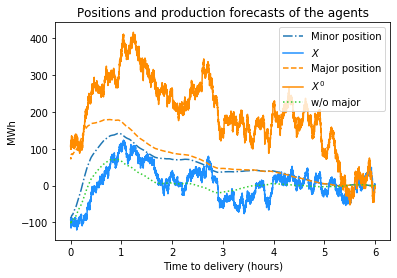}
\end{figure}

We observe that the strategy in the  setting of identical agents and the strategy of the minor player in the
Stackelberg setting converge to the same terminal value due to the terminal penalty. However, in the Stackelberg case, the minor agent position tends to follow the one of the  major player during the first part of the trading period. In the case of identical agents, the fluctuations are not as strong  since, contrary to the case when a major agent is present, the generic minor agent has no incentive to modify her trajectory to follow the leader. During the second half of the trading period, the minor agent position deviates further away from the one of the major agent to target the same terminal position as the mean field in the case of identical agents. We can argue that the strategy of the minor agent becomes more sensitive to the terminal constraint as we get closer to the delivery time: the weight of the terminal constraint in her strategy increases due to the decrease of the instantaneous trading cost.

\paragraph{\textbf{Volatility and price-forecast correlation}}

In this paragraph, we illustrate with simulations the effect of the presence of the major agent on the price characteristics such as the volatility and the correlation between the price and renewable infeed forecasts. The volatility was estimated from simulated price trajectories using a kernel-based non parametric estimator of the instantaneous volatility:
\begin{equation}\label{volestimation}
    \hat \sigma_t^2 = \frac{\sum_{i = 1}^n K_h(t_{i-1}-t)\Delta \tilde P^2_{t_{i-1}}}{\sum_{i = 1}^n K_h(t_{i-1}-t)(t_{i}-t_{i-1})},
\end{equation}
where $K(.)$ is the Epanechnikov kernel: $K(x) = \frac{3}{4} (1-x^2) \mathbbm{1}_{[-1,1]}(x)$
and $K_h(x)= \frac{1}{h}K(\frac{x}{h})$. The parameter $h$ was taken equal to $0.08$ hour ($\approx 5$ minutes).

For a fixed scenario of production forecasts for the minor and major players, drawn in Figure \ref{Fig5}, left graph, we estimated the volatility of the simulated market price for different values of the weights $a_0$ and $a$ assigned, respectively, to the major player and the mean field of minor players in the price impact function. We studied three different combinations of weights to illustrate the impact of the minor players and the major player in the game:  ${a}^0 = {a} = 0.5$ \euro{}/MWh$^{2}$, the impact of the major player and the minor players is the same; ${a}^0 = 0.9, {a} = 0.1$ \euro{}/MWh$^{2}$, the major player has a lot more impact than the minor players, and finally ${a}^0 = 0$ \euro{}/MWh$^{2}$, ${a} = 1$ \euro{}/MWh$^{2}$, equivalent to a market price without major player since she has no market impact in this case. These weights can be seen as the respective market shares held by the major agent and the minor players.

Figure \ref{Fig6}, left graph, shows the estimated volatility trajectories for the three different cases of market shares of the major agent averaged over 1000 simulations. 
We note that the volatility of the market price depends on the strength of impact of the major player: the greater $a
^0$, the higher the volatility. A possible explanation for this phenomenon is that stronger competition in the market (when the major agent is absent or has a small market share) reduces profit opportunities in the market and the agents therefore trade less actively.

For comparison, we also plot in Figure \ref{Fig6}, right graph, the volatility estimated from empirical intraday electricity price data using the same estimator \eqref{volestimation}. This graph is taken from \cite{feron2020price}. We see that the phenomenon of increasing volatility at the approach of the delivery date, clearly visible in the actual electricity markets, is well reproduced by our model.

\begin{figure}[h]
      \centering
      \caption{Left: volatility of simulated prices for different market shares of the major agent. Right: volatility for different delivery hours, estimated empirically  from EPEX  spot intraday market data of January 2017 for the Germany delivery zone.}
      \label{Fig6}
      \includegraphics[width =0.48\textwidth]{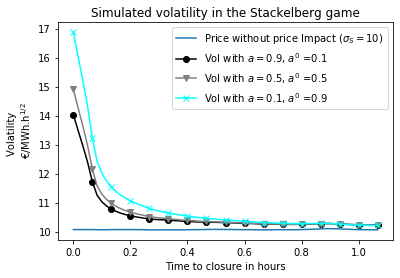}
\includegraphics[width =0.49\textwidth]{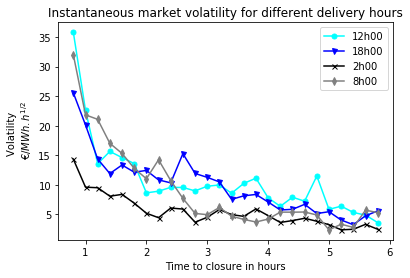}   
\end{figure}

An important stylized feature of intraday market prices, observed empirically in \cite{kiesel2017econometric} and \cite{feron2020price} is the correlation between the price and the renewable production forecasts.
In Figure \ref{Fig7}, we plotted the correlation between the increments of the market price and the increments of the renewable production forecast of the major agent as function of time, in the market impact setting ${a}^0 = {a} = 0.5$ \euro{}/MWh$^{2}$; as well as the correlation between the price increments and the increments of the total aggregate forecast of both the major and minor players. 

The correlation is computed over 15-minutes increments using the following estimator:
$$\hat \rho_{t} = \frac{\sum_{k =1}^{N_{sim}} (\Delta Y^k_t - \overline{\Delta Y}_t)(\Delta P^k_t - \overline{\Delta P}_t)}{\sqrt{\sum_{k = 1}^{N_{sim}} (\Delta Y^k_t - \overline{\Delta Y}_t)^2\sum_{k = 1}^{N_{sim}}(\Delta P^k_t - \overline{\Delta P}_t)^2}},$$
with $N_{sim}$ the number of simulations (we considered $N_{sim} = 50000$) and where $\Delta Y^k_t $ and $\Delta P^k_t = P^{MF,k}_{t+dt}-P^{MF,k}_{t}$ are the increments of, respectively, the forecast process and the market price. 

\begin{figure}[h]
      \centering
      \caption{Correlation between the price increments and the major player renewable production increments v.s the correlation between the price increments and the total renewable production increments}
      \includegraphics[width = 0.48\textwidth]{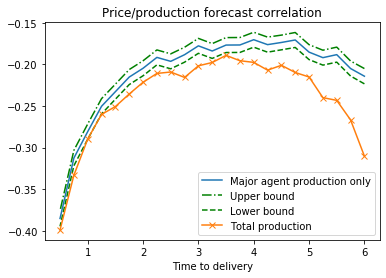}\label{Fig7}
\end{figure}

{For the sake of clarity we only draw the Monte Carlo confidence interval for the case of the  correlation between the major player production and the price considered on Figure \ref{Fig7}. A similar confidence interval was obtained for the case of total production correlation.
In Figure \ref{Fig7}, we observe that the correlation between the production forecast increments of the major agent and the price is lower in absolute value than the correlation between the total production forecast increments and the price. However, the gap between the correlations diminishes as we approach the delivery time.
}
\section*{Acknowledgements}The authors gratefully acknowledge financial support from the ANR
(project EcoREES ANR-19-CE05-0042) and from the
FIME Research Initiative.





\end{document}